\documentclass[11pt]{article}

\usepackage{amsfonts}
\usepackage{amsmath}
\usepackage{amssymb}
\usepackage{amsthm}
\usepackage{psfrag}

\usepackage{graphicx}

\theoremstyle{plain}
\newtheorem{theorem}{Theorem}

\newtheorem{definition}[theorem]{Definition}
\newtheorem{lemma}[theorem]{Lemma}

\theoremstyle{definition}
\newtheorem{example}[theorem]{Example}
\newtheorem{remark}[theorem]{Remark}

\numberwithin{equation}{section}
\numberwithin{theorem}{section}

\DeclareMathOperator{\Int}{int}

\DeclareMathOperator{\bd}{bd}

\DeclareMathOperator{\diag}{diag}

\newcommand{\bs}[1]{{\boldsymbol{#1}}}

\newcommand{\pdt}{\partial_t}
\newcommand{\R}{\mathbf{R}}

\newcommand{\I}{\mathrm{i}}
\newcommand{\D}{\,\mathrm{d}}

\newcommand{\DP}[2]{\langle#1,#2\rangle}

\begin{document}
\title{{Solutions with a bounded support promote permanence of a distributed replicator equation}}

\author{Alexander S. Bratus$^{1,2}$, Vladimir P. Posvyanskii$^{2}$\,, Artem S. Novozhilov$^{{3},}$\footnote{Corresponding author: artem.novozhilov@ndsu.edu}\\[3mm]
\textit{\normalsize $^\textrm{\emph{1}}$Faculty of Computational Mathematics and Cybernetics,}\\[-1mm]
\textit{\normalsize Lomonosov Moscow State University, Moscow 119992, Russia}\\[2mm]
\textit{\normalsize $^\textrm{\emph{2}}$Applied Mathematics--1, Moscow State University of Railway Engineering,}\\[-1mm]\textit{\normalsize Moscow 127994, Russia}\\[2mm]
\textit{\normalsize $^\textrm{\emph{3}}$Department of Mathematics, North Dakota State University, Fargo, ND 58108, USA}}

\date{}

\maketitle

\begin{abstract}
The now classical replicator equation describes a wide variety of biological phenomena, including those in theoretical genetics, evolutionary game theory, or in the theories of the origin of life. Among other questions, the permanence of the replicator equation is well studied in the local, well-mixed case. Inasmuch as the spatial heterogeneities are key to understanding the species coexistence at least in some cases, it is important to supplement the classical theory of the non-distributed replicator equation with a spatially explicit framework. One possible approach, motivated by the porous medium equation, is introduced. It is shown that the solutions to the spatially heterogeneous replicator equation may evolve to equilibrium states that have a bounded support, and, moreover, that  these solutions are of paramount importance for the overall system permanence, which is shown to be a more commonplace phenomenon for the spatially explicit equation if compared with the local model.

\paragraph{\small Keywords:} Replicator equation, reaction--diffusion systems, stability, permanence, uniform persistence
\paragraph{\small AMS Subject Classification:} Primary:  35K57, 35B35, 91A22; Secondary: 92D25
\end{abstract}

\section{Introduction}
The species coexistence is arguably the most important characteristics of a biological (ecological, chemical, etc) system. The clear understanding of this trivial fact led to highly nontrivial theories of mathematical \textit{permanence} \cite{hofbauer1998ega} or \textit{uniform persistence} \cite{smith2011dynamical}, which provide the rigorous framework for the verbal description that ``The presence or the absence of a species is sometimes the point of interest regardless of some variation in their numbers'' \cite{lewontin1969meaning}.

A significant number of results for the species permanence, which mathematically means that the solutions are separated from both zero and infinity, were obtained for the so-called \textit{replicator equation} \cite{hofbauer2003egd,hofbauer1998ega,schuster1983rd}, in the classical form
\begin{equation}\label{eq:0:1}
    \dot w_i=w_i\Bigl(\bigl(\bs{Aw}\bigr)_i-\DP{\bs w}{\bs{Aw}}\Bigr),\quad i=1,\ldots,n,
\end{equation}
where $\bs w=(w_,\ldots,w_n)$ is the vector of frequencies of interacting species, real matrix $\bs A$ describes the interactions in terms of the catalyzing rates, $(\bs u)_i$ is the $i$-th entry of the vector $\bs u$, and $\DP{\cdot}{\cdot}$ is the usual dot product in $\R^n$. Note that if $S_n$ is the standard simplex in $\R^n$ then, due to the normalization term $\DP{\bs w}{\bs{Aw}}$, $\bs w(t)\in S_n$ for any time moment, assuming $\bs w(0)\in S_n$, that is the simplex is invariant with respect to the flow defined by \eqref{eq:0:1}.

Problem \eqref{eq:0:1} is a system of ordinary differential equations (ODE) and therefore describes the dynamics of a well mixed system. It is a common wisdom that the \textit{spatial structure} mediates coexistence \cite{cantrell2003spatial,dieckmann2000}, and therefore we face an important problem to extend, compare, and generalize the results obtained for the local system \eqref{eq:0:1} to the case when we include the spatial variables in our equations. The first fact that should be clearly understood in this respect is that there are different and non-equivalent ways to add the spatial heterogeneity to the model \eqref{eq:0:1}, that is, the results of analysis are \textit{model dependent}.

One way to model the spatial structure is to assume that the individuals are associated with vertices of some graph, and two individuals interact if their vertices are connected by an edge. This approach led to the evolutionary games on graphs (e.g., \cite{lieberman2005evolutionary}). Alternatively, it is also possible to assume that the whole system composed of a number of local populations, within which the infractions are random, and some dispersal rates between the patches are specified (e.g., \cite{schreiber2013spatial}). It is important to remark that in both of these cases the dynamics of the structured populations is different from that of the underlying well-mixed model; in particular, an important phenomenon of cooperation can be maintained in structured populations, opposite to the evolutionary outcomes in local, randomly mixing populations.

One of the most popular ways to add the spatial heterogeneity to the local ODE models is to consider a corresponding reaction-diffusion system, when the Laplace operator, describing the microscopic Browning motion, is added to the rates of the local model. Note, however, that it is incorrect to add the Laplace operator directly to system \eqref{eq:0:1} (see \cite{novozhilov2012reaction} for a review, and \cite{bratus2006ssc,bratus2011,bratus2009existence,cressman1987density,cressman1997sad,hutson1995sst,hutson1992travelling,Vickers1991,vickers1989spa,weinberger1991ssa} for additional details and analysis of special cases). A natural approach to add the spatial heterogeneity through the reaction-diffusion mechanism  to the replicator equation \eqref{eq:0:1} is to start with the equation for the absolute sizes, and not for the frequencies as in \eqref{eq:0:1}, add the Laplace operator, and after this make the change of variables to reduce the system to the problem on simplex (which becomes integral in this case, see below for the exact definition). This idea, which is a mathematical manifestation of the \textit{global regulation} was originally used for the quasispecies model \cite{weinberger1991ssa}, see also \cite{bratus2016diffusive} for more general results, and in \cite{bratus2006ssc} for the hypercycle model; subsequent analysis of the general reaction-diffusion replicator equation was performed in \cite{bratus2011,bratus2009existence}. One of the conclusions that we obtained in the cited works is that the behavior of solutions of the reaction-diffusion replicator equation obtained through the principle of global regulation is qualitatively similar to the solutions of the local model \eqref{eq:0:1}, and in particular the set of all the matrices $\bs A$, for which the system is permanent, is no bigger than this set for model \eqref{eq:0:1}. At the same time the local model \eqref{eq:0:1} is not adequate at least in some cases, as the following example shows.
\begin{example}\label{ex:0:1}Consider an in-vitro system of cooperative RNA replicators, analyzed in \cite{vaidya2012spontaneous}, which can be schematically represented as in Fig. \ref{fig:0:1}.
\begin{figure}[!bh]
\centering
\includegraphics[width=0.45\textwidth]{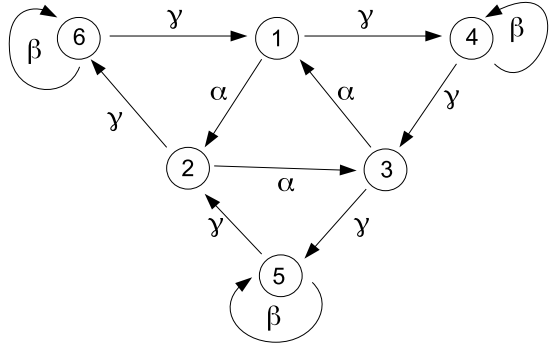}
\caption{A schematic representation of the catalytic network of 6 macromolecules \cite{vaidya2012spontaneous}}\label{fig:0:1}
\end{figure}

It was shown that this particular network of macromolecules is capable of sustaining self-replication (that is, it is permanent, and none of the macromolecules went extinct in experiments). A naive modeling approach would be to consider an interaction matrix
$$
\bs A=    \left[
\begin{array}{cccccc}
 0 & 0 & \alpha  & 0 & 0 & \gamma  \\
 \alpha  & 0 & 0 & 0 & \gamma  & 0 \\
 0 & \alpha  & 0 & \gamma  & 0 & 0 \\
 \gamma  & 0 & 0 & \beta  & 0 & 0 \\
 0 & 0 & \gamma  & 0 & \beta  & 0 \\
 0 & \gamma  & 0 & 0 & 0 & \beta  \\
\end{array}
\right]
$$
and the corresponding replicator equation \eqref{eq:0:1}. It can be shown, however (see \cite{novozhilov2013replicator} for additional details), that this system is not permanent contradicting therefore the experimental results.
\end{example}

Example \ref{ex:0:1} prompts for a modification of the local replicator equation \eqref{eq:0:1} such that the model solution would reflect the permanent nature of the underlying cooperative network. To this end, we suggested in \cite{novozhilov2013replicator} another way to arrive to a reaction-diffusion equation, motivated in significant part by the diffusion equation in the porous medium \cite{barenblatt1989theory,knabner2003numerical}. In particular, for the \textit{absolute sizes} $\bs N=(N_1,\ldots,N_n)$ we can write
$$
\frac{1}{\phi_i(\bs N)}\partial_t N_i=(\bs{Aw})_i+d_i\Delta N_i,\quad i=1,\ldots,n,
$$
where the functions $\phi_i$ specify how the densities affect the diffusion rates, and $d_i$ are some given parameters. In particular, for the simplest case $\phi_i(\bs N)=N_i$ we obtain a quasilinear reaction-diffusion PDE
\begin{equation}\label{eq:0:2}
\partial_t N_i=(\bs{Aw})_iN_i+d_iN_i\Delta N_i,\quad i=1,\ldots,n.
\end{equation}
Equations of this type received much less attention in the literature, compare to the classical reaction-diffusion systems \cite{cantrell2003spatial}, see, e.g., \cite{postnikov2007continuum} for one example, where a simple model was used to model the spread of infection in a population of individuals with low mobilities. Such equations, being quasilinear, pose significant mathematical challenges (e.g., \cite{aronson1986porous,bertsch1990discontinuous,dal1987degenerate}).

In \cite{novozhilov2013replicator} we performed a numerical and analytical analysis of the reaction-diffusion replicator equation, which is obtained from \eqref{eq:0:2} by switching to the vector of frequencies (see below for the exact expressions). In particular, we proved that for sufficiently large parameters $d_i>0$ the equilibria of the reaction--diffusion replicator equation are uniform and coincide with equilibria of \eqref{eq:0:1}, and, more importantly, identified the conditions when their stability properties coincide. Additionally, we found a sufficient condition that the distributed spatially heterogeneous system is permanent. This condition, however, turned out to be more restrictive than that for the local system. Our two most interesting observations were of a numerical nature. We found that 1) for a large number of examples of replicator equations including the spatial structure in the form of equation \eqref{eq:0:2} leads to system permanence even if the original local system does not demonstrate species coexistence, and 2) with time the solutions to the reaction-diffusion replicator equation tend to equilibrium solutions, whose support is only part of the domain in which we consider our problem. Both of these facts can be observed for the problem in Example \ref{ex:0:1}. Here is another simple example to support and illustrate our claims. This example also serves to motivate the subsequent analytical analysis.
\begin{example}\label{ex:0:2}
Consider the replicator equation with the matrix
$$
\bs A=\begin{bmatrix}
        1.1 & 1 \\
        1 & 0 \\
      \end{bmatrix}.
$$
For this example, a straightforward analysis of \eqref{eq:0:1} shows that $w_1(t)\to 1$and $w_2(t)\to 0$ as $t\to\infty$, and therefore the system is clearly not permanent (Fig. \ref{fig:0:2}a). If, however, we consider a reaction--diffusion replicator equation of the form \eqref{eq:0:2}, (the parameters are $d_1=1.1/\pi^2,\,d_2=0.5/\pi^2$), then the numerical experiments show that the spatial heterogeneity stabilizes the system, which becomes permanent (see Fig. \ref{fig:0:2}c,d). Note also that in the long run the solutions concentrate only on a proper subset of the spatial domain $\Omega=(0,1)$ (Fig. \ref{fig:0:2}c).
\begin{figure}[!t]
\centering
\includegraphics[width=0.75\textwidth]{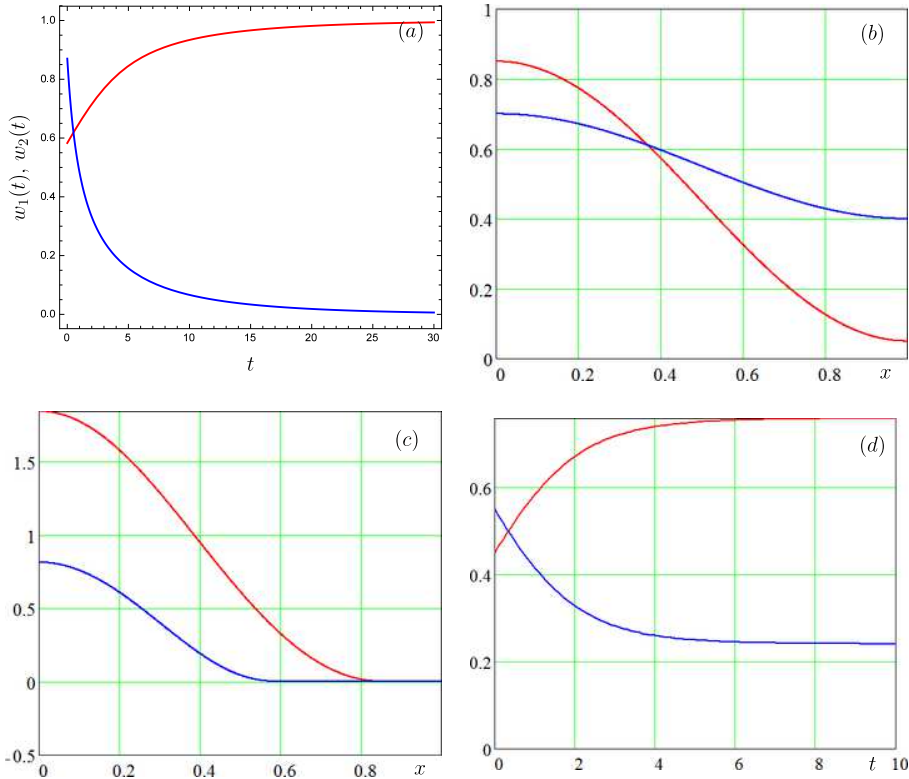}
\caption{Comparison of the local and distributed replicator systems with the same interaction matrix (Example \ref{ex:0:2}). $(a)$ Time dependent solutions to the local system. $(b)$ The initial conditions for the spatially distributed replicator equation. $(c)$ The limit of the solutions to the distributed replicator equation when $t\to\infty$. $(d)$ The time dependent behavior of the integrals of the  solutions in the distributed case}\label{fig:0:2}
\end{figure}
\end{example}
The goal of the present paper is to provide analytical analysis of both observations made in \cite{novozhilov2013replicator} and presented in Example \ref{ex:0:2}, i.e., to study analytically the appearance of solutions that are nonzero only on the part of the spatial domain $\Omega$, in the following we call such situations \textit{solutions with bounded support}, and provide sufficient conditions for the system permanence, which go beyond those valid for the local system. It turns out, as we show, that these two observations are inherently interconnected, and the solutions with bounded support play a significant role in the system permanence.

Before embarking on the analysis of our distributed replicator equation, it is important to mention that the usual definition of the system permanence (see, e.g., \cite{cantrell2003spatial}, the definition of the ``ecological permanence'') requires that $m\leq u_i(\bs x,t)\leq M$ for all $t>t_0$ and all $\bs x\in\Omega$, $m,M$ are given constants. In view of the special solutions we are about to study (Fig. \ref{fig:0:2}c) this definition is clearly not satisfactory for us. Therefore, in the rest of the paper the term ``permanence'' means that the integral value of the variables, i.e.,
$$
\int_\Omega u_i(\bs x,t)\D \bs x
$$
is separated from 0 for any time $t>t_0$ (see below precise Definition \ref{def1:1}).

The rest of the paper is organized as follows. In Section \ref{sec:2} we collect necessary notations and introduce a key definition of the \textit{resonant parameters}. In Section \ref{sec:3} we show how the solutions with bounded support naturally appear in our problems. Section \ref{sec:4} is devoted to the study of the connections between the solutions with bounded support and system permanence. In Appendix we prove some auxiliary facts.

\section{Model statement}\label{sec:2}
In this section we collect necessary notations and facts required for the subsequent analysis.

Let $\Omega$ be a bounded domain in $\R^m$, where $m$ is equal to 1, 2, or 3, depending on the required geometry, with a piecewise-smooth boundary $\Gamma$, $\bs A=(a_{ij})_{n\times n}$ a given real matrix, $\bs v=\bs v(\bs x,t)=\bigl(v_1(\bs x,t),\ldots,\bs v_n(\bs x,t)\bigr)$ a vector-function, $\bs x\in \Omega$, $t\geq 0$. We introduce the notations
\begin{align*}
\bigl(\bs A\bs v\bigr)_k&=\sum_{j=1}^n a_{kj}v_j(\bs x,t),\quad k=1,\ldots,n,\\
\DP{\bs{Av}}{\bs v}&=\sum_{j=1}^n\bigl(\bs{Av}\bigr)_jv_j(\bs x,t)=\sum_{j,k=1}^na_{kj}v_k(\bs x,t)v_j(\bs x,t).
\end{align*}

We consider the initial-boundary value problem (here $d_k>0,\,k=1,\ldots,n$ are parameters)
\begin{equation}\label{eq1:1}
    \pdt v_k=v_k\Bigl(\bigl(\bs A\bs v\bigr)_k-f^s(\bs v)+d_k\Delta v_k\Bigr),\quad k=1,\ldots,n,
\end{equation}
with the initial and boundary conditions
\begin{equation}\label{eq1:2}
    v_k(\bs x,0)=\varphi_k(\bs x),\quad \left.\frac{\partial v_k(\bs x,t)}{\partial\nu}\right|_{\bs x\in \Gamma}=0,\quad k=1,\ldots,n,
\end{equation}
where $\nu$ is the outward normal to $\Gamma$. In the system \eqref{eq1:1} we have
\begin{equation}\label{eq1:3}
    f^s(\bs v)=\int_{\Omega}\Bigl(\DP{\bs{Av}}{\bs v}-\sum_{k=1}^n d_k\|\nabla v_k\|^2\Bigr)\D \bs x.
\end{equation}
At this point we would like to remark that the system \eqref{eq1:1}--\eqref{eq1:3} is not a classical system of partial differential equations (PDE) since $f^s(\bs v)$ is a functional on the solutions to the problem  \eqref{eq1:1}--\eqref{eq1:2}.

From \eqref{eq1:1}--\eqref{eq1:3} it follows that
$$
\frac{\D}{\D t}\left(\sum_{k=1}^n\int_{\Omega} v_k(\bs x,t)\D \bs x\right)=0,
$$
which means that
\begin{equation}\label{eq1:4}
    \sum_{k=1}^n\int_{\Omega} v_k(\bs x,t)\D \bs x=\text{const}
\end{equation}
for any $t\geq 0$, where the constant can be chosen arbitrarily, we set it equal to one. This means that the integral simplex (see below) of the problem \eqref{eq1:1}--\eqref{eq1:3} is invariant.

Problem \eqref{eq1:1}--\eqref{eq1:2} is a spatially explicit replicator equation of the reaction-diffusion type and describes, for instance, the population dynamics of self-replicating and interacting molecules. In this interpretation $v_k(\bs x,t)$ is the relative density of the macromolecules of the $k$-th type relative to the total density in the domain $\Omega$ at the time moment $t$. The functional $f^s(\bs v)$ is hence the mean population fitness, and the expression $\bigl(\bs A\bs v\bigr)_k$ is the fitness of the $k$-th type of macromolecules at the point $\bs x\in\Omega$ at the time moment $t$.

From the physical meaning of the problem we conclude that the solutions to \eqref{eq1:1}--\eqref{eq1:3} should be sought among the set of non-negative functions $v_k(\bs x,t)\geq 0,\,\bs x\in\Omega,\,t\geq 0,k=1,\ldots,n$.
In the following we assume that the functions $v_k(\bs x,t),\,\bs x\in\Omega,\,t\geq 0,k=1,\ldots,n$ are smooth with respect to $t$ and, together with their derivatives with respect to $t$, belong to the Sobolev space  $W^{1,2}$, if $m=1$, and to $W^{2,2}$, if $m=2,3$, for each fixed $t>0$. Here $W^{s,2}$ is the space of square integrable functions in $\Omega$ together with their (weak) derivatives up to the order  $s$. We note that from the embedding theorems (e.g., \cite{evans_2010}) it follows that such functions coincide with continuous functions almost everywhere in $\Omega$.

Denote $\Omega_t=\Omega\times[0,\infty)$ and consider the set of functions $B(\Omega_t)$ with the norm
$$
\|z(\bs x,t)\|_{B(\Omega_t)}=\max_{t\geq 0}\left\{\|z(\bs x,t)\|_{W^{s,2}}+\|\pdt z(\bs x,t)\|_{W^{s,2}}\right\},\quad s=1,2.
$$
Denote $S_n(\Omega_t)$ the set of non-negative functions $\bs v(\bs x,t)=\bigl(v_1(\bs x,t),\ldots,v_n(\bs x,t)\bigr)$ such that $v_k(\bs x,t)\in B(\Omega_t)$ for all $k$ and satisfy \eqref{eq1:4} with the constant equal to one:
\begin{equation}\label{eq1:5}
    \sum_{k=1}^n\int_{\Omega} v_k(\bs x,t)\D \bs x=1.
\end{equation}
The set $S_n(\Omega_t)$ is the \textit{integral simplex} in the space of vector-functions, each component of which belongs to $B(\Omega_t)$.

The boundary elements (denoted $\bd S_n(\Omega_t)$) of the integral simplex $S_n(\Omega_t)$ are the vector-functions $\bs v(\bs x,t)=\bigl(v_1(\bs x,t),\ldots,v_n(\bs x,t)\bigr)$ such that for a non empty set of indexes  $K_0\subset \{1,\ldots,n\}$
$$
\overline{v}_k(t)=\int_\Omega v_k(\bs x,t)\D \bs x=0,\quad k\in K_0,
$$
and $\overline{v}_k(t)>0$, $k\notin K_0$, $t\geq 0$. Due to the simplex invariance
\begin{equation}\label{eq1:6}
    \sum_{k\notin K_0}\overline{v}_k(t)=1.
\end{equation}
The interior elements of the simplex $S_n(\Omega_t)$ (denoted $\Int S_n(\Omega_t)$) are the vector-functions $\bs v(\bs x,t)\in S_n(\Omega_t)$, for which
$$
\overline{v}_k(t)=\int_\Omega v_k(\bs x,t)\D \bs x>0,\quad k=1,\ldots,n,\quad t\geq 0.
$$
Furthermore, without loss of generality we assume that the measure of $\Omega$ is equal to 1, $|\Omega|=1$.

\begin{remark}Since $v_k(\bs x,t)\in W^{s,2},\,k=1,\ldots,n$ for $s=1$ or $s=2$ for each $t\geq 0$, then from the embedding theorems it follows that they coincide almost everywhere with continuous functions. Therefore, taking into account non-negativity of the functions, we conclude that if the mean integral value $\overline{v}_k(t)=0$ then $v_k(\bs x,t)=0$ almost everywhere in $\Omega$. Therefore the set $\bd S_n(\Omega_t)$ consists of vector-fucntions for which
$$
v_k(\bs x,t)=0,\quad k\in K_0,
$$
and also equality \eqref{eq1:6} holds.
\end{remark}

We consider weak solutions to \eqref{eq1:1}--\eqref{eq1:3}. Vector function $\bs v(\bs x,t)\in S_n(\Omega_t)$ is a weak solution if the following integral identity holds:
\begin{align*}
\int_0^\infty\int_\Omega \pdt v_k(\bs x,t)\eta(\bs x,t)\D \bs x\D t&=\int_0^\infty\int_\Omega v_k(\bs x,t)\Bigl(\bigl(\bs{Av}\bigr)_k-f^s(\bs v)\Bigr)\eta(\bs x,t)\D \bs x\D t\\
&-d_k\int_0^\infty\int_\Omega \DP{\nabla v_k(\bs x,t)}{\nabla \eta(\bs x,t)}\D \bs x\D t,
\end{align*}
for each function $\eta(\bs x,t)$, which for $\bs x\in\Omega$ is differential with respect to $t$ has a compact support for each fixed $t\in[0,\infty)$, and also for any $t\geq 0$ belongs to $W^{s,2}(\Omega)$ for $s=1$ or $s=2$.

Together with the problem \eqref{eq1:1}--\eqref{eq1:3} we also consider a system of ordinary differential equations which can be obtained formally from the original one when $d_k\to 0$:
\begin{equation}\label{eq1:7}
    \dot w_k=w_k\Bigl(\bigl(\bs{Aw}\bigr)_k-f^l(\bs w)\Bigr),\quad k=1,\ldots n,
\end{equation}
with the initial conditions
$$
w_k(0)=w_k^0,\quad k=1,\ldots,n.
$$
Here
$$
f^l(\bs w)=\DP{\bs{Aw}}{\bs w}=\sum_{i,j=1}^n a_{ij}w_iw_j.
$$
Problem \eqref{eq1:7} is considered on the set on non-negative vector-functions $\bs w(t)=\bigl(w_1(t),\ldots,w_n(t)\bigr)$, which for each time moment $t$ belong to the standard simples $S_n$, i.e.,
\begin{equation}\label{eq1:8}
    \sum_{k=1}^n w_k(t)=1,\quad w_k(t)\geq0,\quad k=1,\ldots,n.
\end{equation}
In analogy with the boundary and interior sets for the integral simplex we denote the boundary $\bd S_n$ (there is at least one $k$ such that $w_k(t)=0$) and the interior set $\Int S_n$ (for all $k=1,\ldots, n$, $w_k(t)>0$ for all  $t$). The sets $\bd S_n$ and $\Int S_n$ are invariant.
\begin{remark}\label{remark:2} For each element $\bs v(\bs x,t)\in S_n(\Omega_t)$ we can identify the element $\bs w(t)\in S_n$, if we set $\bs w(t)=\overline{\bs v}(t)$, where, and everywhere else in the text, the bar denotes the mean integral value through $\Omega$:
$$
\overline{v}_k(t)=\int_\Omega v_k(\bs x,t)\D \bs x,\quad k=1,\ldots,n,\quad t\geq 0.
$$
\end{remark}

Due to the reasons discussed informally in Introduction we use the following
\begin{definition}\label{def1:1}The system \eqref{eq1:1}--\eqref{eq1:3} is called permanent, if there are such $\varepsilon_0>0$ and $\delta_0>0$ that for all components $\bs v(\bs x,t)\in S_n(\Omega_t)$ of the system \eqref{eq1:1}--\eqref{eq1:3} it holds that
$$
\liminf_{t\to\infty}\|v_k(\bs x,t)\|\geq \varepsilon_0>0,\quad k=1,\ldots,n,
$$
if the initial conditions \eqref{eq1:2} satisfy
$$
\|\varphi_k(\bs x)\|\geq \delta_0>0.
$$
\end{definition}

Here and below $\|\cdot\|$ denotes the norm in $L^2(\Omega)$.

A number of necessary and sufficient conditions of permanence of \eqref{eq1:7} is given in \cite{hofbauer1998ega}. The one which we will use says that the system \eqref{eq1:7} is permanent if
\begin{equation}\label{eq1:9}
    \DP{\bs{Aw}}{\bs p}-\DP{\bs{Aw}}{\bs w}>0
\end{equation}
for any equilibria $\bs w\in\bd S_n$. Here $\bs p$ is some fixed point in $\Int S_n$, i.e.,
$$
\sum_{k=1}^n p_k=1,\quad p_k>0,\quad k=1,\ldots,n.
$$
In \cite{novozhilov2013replicator} we showed that a similar condition can be obtained for the distributed system \eqref{eq1:1}--\eqref{eq1:3}. This condition, opposite to \eqref{eq1:9}, must be checked on \textit{all} elements $\bs w\in \bd S_n$ and at the same time the following condition on the parameters $d_k$ must be true:
\begin{equation}\label{eq1:10}
    \lambda_1d_{\min}>\mu,
\end{equation}
where $d_{\min}=\min_k\{d_k\}$, $\mu$ is the spectral radius of $\bs A$, and $\lambda_1$ is the first nonzero eigenvalue of the boundary value problem
\begin{equation}\label{eq1:11}
    -\Delta \psi (\bs x)=\lambda \psi(\bs x),\quad \bs x\in \Omega, \quad \left.\frac{\partial \psi}{\partial \nu}\right|_{\bs x\in \Gamma}=0.
\end{equation}

In \cite{novozhilov2013replicator} we proved that if the condition \eqref{eq1:10} holds then all the equilibria of \eqref{eq1:1}--\eqref{eq1:3} coincide with the equilibria of the local system \eqref{eq1:7}, i.e., they are spatially homogeneous. Hence, the conditions \eqref{eq1:9} and \eqref{eq1:10} provide the system permanence only if the equilibria are spatially homogeneous. This observation implies that the analysis we presented in \cite{novozhilov2013replicator} cannot rigorously identify the cases such that the spatial structure would stabilize the system. Therefore it is important to consider the case when  \eqref{eq1:10} does not hold.

\begin{definition}\label{def:2:4} We shall say that the set of parameters $\{d_i\}_{i=1}^n$ of the distributed system \eqref{eq1:1}--\eqref{eq1:3} is resonant if there exists the eigenvalue $\lambda$ of the problem \eqref{eq1:11}, such that
\begin{equation}\label{eq1:12}
\det (\bs A-\lambda \bs D)=0,
\end{equation}
where $\bs D=\diag(d_1,\ldots,d_n)$.
\end{definition}

For the following we assume that matrix $\bs A$ is non-negative and primitive such that the conditions of the Frobenius-Perron theorem hold. We denote $\mu$ the dominant eigenvalue of $\bs A$ to which corresponds a positive eigenvector. If $\bs D=d_0\bs I$, then from \eqref{eq1:12} we have that the maximal value of $d_0$ for which \eqref{eq1:12} holds is
\begin{equation}\label{eq1:13}
    d_0=\frac{\mu}{\lambda_1}\,,
\end{equation}
where $\lambda_1$ is the smallest positive eigenvalue of \eqref{eq1:11}.

\section{Spatially inhomogeneous solutions. Solutions with bounded support}\label{sec:3}
The stationary solutions to the problem \eqref{eq1:1}--\eqref{eq1:3} satisfy the system
\begin{equation}\label{eq3:1a}
u_i\Bigl((\bs{Au})_i-\bar f^{s}+d_i\Delta u_i\Bigr)=0,\quad i=1,\ldots,n,\quad \bs x\in\Omega,
\end{equation}
with the boundary conditions
\begin{equation}\label{eq3:1b}
    \frac{\partial u_i}{\partial\bs\nu}|_{\bs x\in\Gamma}=0,\quad i=1,\ldots,n.
\end{equation}
Here
\begin{equation}\label{eq3:2}
    \bar f^s=\int_{\Omega}\Bigl(\DP{\bs{Au}}{\bs u}-\sum_{i=1}^n d_i\|\nabla u_i\|^2\Bigr)\D \bs x.
\end{equation}
Together with problem  \eqref{eq3:1a}--\eqref{eq3:2} consider the equations for the equilibria of the local replicator equation
\begin{equation}\label{eq3:3}
    w_i\Bigl((\bs{Aw})_i-\DP{\bs{Aw}}{\bs w}\Bigr)=0,\quad \bs w\in S_n,\quad i=1,\ldots,n.
\end{equation}
Solutions to \eqref{eq3:3} are denoted below $\bs{\hat{w}}=(\hat w_1,\ldots,\hat w_n)$.

In \cite{novozhilov2013replicator} we showed that if the the set $\{d_i\}_{i=1}^n$ is not resonant then all the stationary solutions to \eqref{eq1:1}, \eqref{eq1:2} coincide with equilibria $\bs{\hat w}$ of \eqref{eq3:3}. Here our first goal is to show that existence of the resonant parameters implies the existence of spatially inhomogeneous stationary solutions of the distributed replicator system.

\begin{theorem}\label{th3:1}
Let the set $\{d_i\}$ be resonant with respect to some eigenvalue $\lambda_s$ of problem \eqref{eq1:11}. Then there exist nonnegative spatially inhomogeneous solutions to the system \eqref{eq3:1a}--\eqref{eq3:2} of the form
\begin{equation}\label{eq3:4}
    \bs{u}(\bs x)=\bs{\hat w}+m\bs c^s\psi_s(\bs x),
\end{equation}
where $\bs{\hat w}$ solves \eqref{eq3:3}, $m$ is an arbitrary constant, $\bs c^s=(c_1^s,\ldots,c_n^s)$ is a fixed vector, and $\psi_s$ is the eigenfunction of \eqref{eq1:11} corresponding to $\lambda_s$.
\end{theorem}
\begin{proof}We will look for a solution to \eqref{eq3:1a}--\eqref{eq3:2} in the form
$$
u_k(\bs x)=\hat w_k+U_k(\bs x),\quad k=1,2,\ldots,n,
$$
where
\begin{equation}\label{eq3:5}
    U_k(\bs x)=\sum_{s=1}^\infty c_s^k\psi_s(\bs x).
\end{equation}
This is possible since the eigenfunctions form a complete system.

We have
\begin{align*}
(\bs{AU})_k&-\bar f^{s}_{\bs U}+d_k\Delta U_k(\bs x)=0,\quad k=1,\ldots,n,\quad \bs x\in\Omega,\\
\bar f^s_{\bs U}&=\int_{\Omega}\Bigl(\DP{\bs{AU}}{\bs U}-\sum_{i=1}^n d_i\|\nabla U_i\|^2\Bigr)\D \bs x,
\end{align*}
where $\bs U=(U_1,\ldots,U_n)$.

Taking the inner products in the last equality with $\psi_s$ consecutively and taking into account the orthogonality and normalization of the eigenfunctions implies
$$
(\bs A-\lambda_k\bs D) \bs c^k=0,\quad k=1,2,\ldots,s,\ldots.
$$
Due to the fact that the set $\{d_i\}$ is resonant for $k=s$ we have that $\det(\bs A-\lambda_s\bs D)=0$, and hence the $s$-th system has a nontrivial solution. Therefore we conclude that there is a stationary solution in the form \eqref{eq3:4}.

Without loss of generality we can take $\DP{\bs c^k}{\bs c^k}=1$. By choosing the arbitrary constant $m$ such that
$$
(\min_{\bs x\in\Omega} \phi_s(\bs x))mc^s_i+\hat w_i\geq 0,\quad i=1,\ldots,n,
$$
we guarantee that the found solutions are non-negative, which concludes the proof.
\end{proof}

In some cases the spatially heterogeneous solutions, whose existence was proved in Theorem \ref{th3:1}, can be found explicitly, as the following example shows.

\begin{example}Consider the stationary solutions of the distributed hypercycle equation in the spatial domain $\Omega=(0,1)$. This means that $(\bs{Au})_i=a_iu_{i-1},\,i=1,\ldots,n,\,u_{0}:=u_n$.

We have
$$
d_i\frac{\D^2 u_i}{\D x^2}+a_i u_{i-1}-\bar f^s=0,\quad i=1,\ldots,n.
$$
Let us look for the solution in the form
$$
u_i(x)=\frac{\bar f^s}{a_{i+1}}+v_i(x),\quad i=1,\ldots,n.
$$
Then
$$
\frac{\D^2 v_1}{\D x^2}+\frac{a_1}{d_1}v_n=0,\quad \frac{\D^2 v_n}{\D x^2}+\frac{a_n}{d_n}v_{n-1}=0,
$$
which implies
$$
\frac{\D^4 v_1}{\D x^4}-\frac{a_1a_n}{d_1d_n}v_{n-1}=0.
$$
We can continue and finally obtain
$$
\frac{\D^{2n} v_1}{\D x^{2n}}+(-1)^{n+1}R_nv_1=0,\quad R_n=\prod_{i=1}^n\frac{a_i}{d_i}\,.
$$
Assume that the set $\{d_i\}$ is resonant with respect to the first eigenvalue $\lambda_1=\pi^2$ of the problem \eqref{eq1:11} on $\Omega=(0,1)$. Equation \eqref{eq1:12} implies that the set $\{d_i\}$ will be resonant if
\begin{equation}\label{eq1:16}
\prod_{i=1}^nd_i=\frac{\prod_{i=1}^n a_i}{\lambda_1^n}\,.
\end{equation}
We note that in the special case $d_i=d_0,\,i=1,\ldots,$ \eqref{eq1:16} turns into
$$
d_0=\frac{\left(\prod_{i=1}^na_i\right)^{1/n}}{\lambda_1}\,.
$$

From \eqref{eq1:16} it follows that
$$
\pi^2=R_n^{1/n}.
$$
This means that the characteristic polynomial of the differential equation has a pair of imaginary roots $\pm\pi\I$. Taking into account the boundary conditions \eqref{eq1:2} we get
$$
v_1(x)=m\cos \pi x,\quad 0<x<1,
$$
where $m$ is a constant. Then
$$
v_n(x)=m\pi^2\frac{d_1}{a_1}\cos \pi x,\quad v_2(x)=m\pi^{2(n-1)}\frac{d_2\cdots d_n}{a_2\cdots a_n}\cos \pi x.
$$
Finally, from the equality for $\bar f^s$ one has
$$
\bar f^s=\frac{1}{\gamma}\,,\quad \gamma=\sum_{i=1}^n\frac{1}{a_i}\,.
$$
We can always choose the constant $m$ such that the found solutions for $u_i(x)$ are nonnegative, and hence we found spatially heterogeneous stationary solutions to the distributed hypercycle system.
\end{example}

Let again $\lambda_1$ be the first nonzero eigenvalue of the eigenproblem \eqref{eq1:11} and $\bs D^1=\diag(d_1^1,\ldots,d_n^1)$ be the corresponding set of the resonant parameters, that is we assume that
\begin{equation}\label{eq3:6}
\det(\bs A-\lambda_1\bs D^1)=0.
\end{equation}

Consider another set of parameters $\bs D=\diag(d_1,\ldots,d_n)$ such that
\begin{equation}\label{eq3:7}
    d_i=\delta d_i^1,\quad 0<\delta<1,\quad i=1,\ldots,n.
\end{equation}
Now it follows from simple arguments (see Lemma \ref{lem:A:1}) that the equality \eqref{eq3:6} with a new matrix $\bs D$ should be true for some new $\lambda>\lambda_1$. Since the spectrum of the problem \eqref{eq1:11} is discrete then for $\delta$ close enough to $1$ we get
$$
\det(\bs A-\lambda_1\bs D)\neq 0.
$$
A natural question to ask is what actually happens with the solutions to \eqref{eq3:1a}-\eqref{eq3:2} in this case. An answer is provided by the following theorem, which shows how the changes in the parameters $\bs D$ yield non homogeneous stationary solutions with a bounded support.

\begin{theorem}\label{th3:2}Let $\Omega=(0,1)$ and let condition \eqref{eq3:6} hold for the resonant set $\bs D^1$. Then for any $\bs D$ satisfying \eqref{eq3:7} there is such $0<l<1$ for which there exist spatially heterogeneous solutions to the stationary problem \eqref{eq3:1a}--\eqref{eq3:2} with the support $\Omega_l=(0,l)$.
\end{theorem}
\begin{proof}The eigenvalues and eigenfunctions of \eqref{eq1:11} for $\Omega_l$ are
$$
\lambda_0^l=0,\quad \lambda_k^l=\left(\frac{k\pi}{l}\right)^2,\,k=1,2,\ldots,\quad \phi_0^l(x)=\frac{1}{l}\,,\phi_k^l(x)=\sqrt{\frac{2}{l}}\cos\frac{\pi x}{l}\,.
$$
Hence the eigenvalues for $\Omega_l$ are greater then those for $\Omega$ and hence, due to the continuous dependence of $\lambda_k^l$ on $l$ there will be $0<l<1$ such that
\begin{equation}\label{eq3:9}
    \det(\bs A-\lambda_1^l\bs D)=0.
\end{equation}
We look for the solutions to \eqref{eq3:1a}--\eqref{eq3:2} in the form
$$
v_i(x)=\frac{\hat w_i}{l}+U_i^l(x),\quad U_i^l(x)=\sum_{s=1}^\infty c_s^l\psi_s^l(x).
$$
Reasoning similarly to the proof of Theorem \ref{th3:1} we find that for $\bs c^1=(c_1^1,\ldots,c_n^1)$ we have
$$
(\bs A-\lambda_1^l\bs D)\bs c^1=0.
$$
Due to \eqref{eq3:9} this system has a nontrivial solution, which we can normalize as $\DP{\bs c^1}{\bs c^1}=1$.

The found solutions can be represented for $x\in\Omega_l$ as
$$
v_i(x)=\frac{\hat w_i}{l}\left(1+mc_i^1\cos \frac{\pi x}{l}\right),
$$
where $m$ is an arbitrary constant. Let $m=|c^1_k|^{-1}$, where $|c_k^1|=\max\{|c^1_1|,\ldots,|c^1_n|\}$ and consider the functions
\begin{equation}\label{eq3:8}
\begin{split}
  u_i(x) &=\begin{cases}
            \frac{\hat w_i}{l}\left(1+\frac{c^1_i}{|c^1_k|}\cos \frac{\pi x}{l}\right),&0<x<l,\\
            \frac{\hat w_i}{l}\left(1-\frac{c^1_i}{|c^1_k|}\right),&l\leq x\leq 1,
           \end{cases},\quad i\neq k, \\
  u_k(x) &=\begin{cases}
            \frac{\hat w_k}{l}\left(1+\cos \frac{\pi x}{l}\right),&0<x<l,\\
            0,&l\leq x\leq 1.
           \end{cases}
\end{split}
\end{equation}
By construction functions $u_i$ are continuous  together with their derivatives at $x=l$, hence the obtained solutions are in $W^{1,2}(\Omega)$. Moreover, ${\rm{supp}}\,u_k=\Omega_l\subset\Omega$.
\end{proof}

The following example illustrates Theorem \ref{th3:2}.
\begin{example}\label{ex3:3} Consider the stationary solutions to the hypercyclic system in the particular case $d_i=d_0,a_i=a_0,i=1,2,\ldots,n$. Assume that
$$
d_0<\frac{a_0}{\pi^2}\,.
$$
Clearly there exists $0<l<1$ such that
$$
d_0=\frac{a_0}{\lambda_1^l}\,,\quad \lambda_1^l=\left(\frac{\pi}{l}\right)^2.
$$
The corresponding stationary solutions with the support given by $\Omega_l$ are
$$
u_i(x)=\begin{cases}
            \frac{a_0}{nl}\left(1+\cos \frac{\pi x}{l}\right),&0<x<l,\\
            0,&l\leq x\leq 1,
           \end{cases},\quad i=1,\ldots,n.
$$
In this case
$$
l=\pi\sqrt{\frac{d_0}{a_0}}\,.
$$
\end{example}

It can be directly checked that the results of Theorem \ref{th3:2} and Example \ref{ex3:3} can be explicitly generalized on some other domains in $\R^2$ or $\R^3$.

\begin{enumerate}
\item We can consider the square $\Omega=(0,1)\times(0,1)$ and a rectangle $\Omega_{l_1,l_2}=(0,l_1)\times(0,l_2)$. In this case the corresponding stationary solutions, nonzero only on $\Omega_{l_1,l_2}$, have the form
    $$
 u_i(x)=
\begin{cases}
            \frac{\hat w_i}{l_1l_2}\left(1+\cos \frac{\pi x_1}{l_1}\right)\left(1+\cos \frac{\pi x_2}{l_2}\right),&0<x_1<l_1,\,0<x_2<l_2,\\
            0,&l_1\leq x_1\leq 1,\,l_2\leq x_2\leq 1.
           \end{cases}
$$
\item In the case of the circle $\Omega=\{(x_1,x_2)\colon x_1^2+x_2^2<1\}$ for the solutions that do not depend on the polar angle
$$
u_i(r)=\begin{cases}
            \frac{\hat w_i}{\pi^2 l}\left(1-c_0J_0\left(\frac{\mu_1^1r}{l}\right)\right),&0<r<l,\\
            0,&l\leq r\leq 1.
           \end{cases}
$$
Here $r=\sqrt{x_1^2+x_2^2}$, $J_0$ is Bessel's function of the first kind, $\mu_1^1$ is the first positive zero of $J_1(r)$, $c_0=(J_0(\mu_1^1))^{-1}$.
\item In the case of the sphere $\Omega={(x_1,x_2,x_3)\colon x_1^2+x_2^2+x_3^2<1}$ for the solutions that do not depend of angular variables, we find
$$
u_i(r)=\begin{cases}
            \frac{3\hat w_i}{4\pi l^3}\left(1-c_1\sqrt{\frac{2l}{\pi\mu r}\sin\frac{\mu r}{l}}\right),&0<r<l,\\
            0,&l\leq r\leq 1,
           \end{cases}
$$
where $\mu$ is the first positive root of the equation $\tan r=2 r$,
$$
c_1=\left(\sqrt{\frac{2}{\pi \mu}}\sin \mu\right)^{-1},\quad r=\sqrt{x_1^2+x_2^2+x_3^2}.
$$
\end{enumerate}

The list of examples can be extended. Which is more important, however, is that in the general case we can conclude that the stationary spatially heterogeneous solutions with a support, which is a proper subset of $\Omega$, appear if there exists a domain $\Omega_1\subset \Omega$ with a smooth boundary $\Gamma_1$ such that the linear combination of the solutions to the eigenvalue problem \eqref{eq1:11} in $\Omega_1$ allows a continuous extension into the domain $\Omega\setminus \Omega_1$ in the case $\dim \Omega=1$ and a smooth extension in this domain in the case $\dim \Omega=2$ or $3$. From the variational principle (e.g., \cite{cantrell2003spatial}) it follows that the eigenvalues of \eqref{eq1:11} in $\Omega_1$ are bigger than the eigenvalues of the same problem solved in $\Omega$. Moreover, if the measure of $\Omega_1$ decreases the eigenvalues will grow.

We remark that using similar to Theorem \ref{th3:2} reasonings it is possible to show more, in particular that at least some of these spatially heterogeneous solution, which existence was proved in Theorem \ref{th3:2}, are attracting. Here is an illustration by an explicit example.
\begin{example}Consider a hypercyclic system on $\Omega=(0,1)$ with the matrix
$$
\bs A=\begin{bmatrix}
        0 & k \\
        k & 0 \\
      \end{bmatrix}
$$
and assume that the parameters in \eqref{eq1:1}--\eqref{eq1:3} are
$$
d_1=d_2=d=\frac{k}{(2\pi)^2}\,.
$$
This means that the parameters are resonant with the second nonzero eigenvalue of \eqref{eq1:11}. Let us look for a solution in the form
\begin{equation}\label{eq3:15}
    v_i(x,t)=\begin{cases}g_i(t)(1+\cos 2\pi x),&0<x<1/2,\\
    0,&1/2\leq x<2,
    \end{cases}
\end{equation}
$i=1,2$. In this case
$$
f^{s}(t)=\frac{k}{2}\Bigl(3g_1(t)g_2(t)-\frac 12\bigl(g_1^2(t)+g_2^2(t)\bigr)\Bigr).
$$
From the condition \eqref{eq1:5} it follows that
\begin{equation}\label{eq3:16}
    g_1(t)+g_2(t)=2.
\end{equation}
On integrating the system of the equations through $\Omega$ and taking into account \eqref{eq3:15}, we obtain the ODE system
\begin{align*}
\dot g_1&=\frac{k}{2}\left(g_2^2-g_1^2\right),\\
\dot g_2&=\frac{k}{2}\left(g_1^2-g_2^2\right).
\end{align*}
Using \eqref{eq3:16} yields
$$
\dot g_i=2k(1-g_i),
$$
and hence
$$
g_i(t)=1+(g_i(0)-1)e^{-2kt},
$$
therefore
$$
g_i(t)\to 1,\quad i=1,2,\,t\to\infty,
$$
and therefore all the solutions of this particular form will tend to the equilibrium solution with a bounded support, namely
\begin{equation*}
    \lim_{t\to\infty} v_i(x,t)=u_i(x)=\begin{cases}1+\cos 2\pi x,&0<x<1/2,\\
    0,&1/2\leq x<2.
    \end{cases}
\end{equation*}

\end{example}
\section{Sufficient conditions for permanence}\label{sec:4}
One of the possible sufficient conditions for the ODE replicator equation \eqref{eq1:7} to be permanent takes the following form. If there exists a $\bs p\in\Int S_n$ such that
\begin{equation}\label{eq4:1}
    \DP{\bs p}{\bs{Aw}}>\DP{\bs{Aw}}{\bs w}
\end{equation}
for all equilibria $\bs w\in\bd S_n$ then system \eqref{eq1:7} is permanent. From Remark \ref{remark:2} it follows that we can identify any function $\bs v(\bs x,t)\in S_n(\Omega)$ with an element $\bs w(t)\in S_n$, by having $\bs w(t)=\overline{\bs v}(t)$, and the same is true for the elements on $\bd S_n(\Omega)$. Therefore one can expect that an analogous to \eqref{eq4:1} condition for the distributed replicator system may look
\begin{equation}\label{eq4:2}
    \DP{\bs p}{\bs{Aw}}>\DP{\bs{Aw}}{\bs w},\quad \bs w(t)=\overline{\bs v}(t)\in\bd S_n,\quad \bs p\in \Int S_n.
\end{equation}
This is indeed true, however, we will show that system \eqref{eq1:1}--\eqref{eq1:3} can be permanent even in a situation when the condition \eqref{eq4:2} does not hold.

First we formulate and prove an auxiliary lemma.
\begin{lemma}\label{lem4:1}Let the set of parameters $\{d_k\}$ of system \eqref{eq1:1}--\eqref{eq1:3} be resonant with respect to the first nonzero eigenvalue of \eqref{eq1:11} in $\Omega$. If there exist spatially nonhomogeneous solutions to \eqref{eq1:1}--\eqref{eq1:3}
\begin{equation}\label{eq4:3}
    v_k(\bs x,t)=w_k(t)+V_k(\bs x,t),\quad w_k(t)=\int_{\Omega_1}v_k(\bs x,t)\D \bs x,\quad k=1,\ldots,n,
\end{equation}
with the support in $\Omega_1\subset \Omega$, with the measure $S_1$, such that at least one spatially nonhomogeneous component (with the index $K_0$) of the solution  satisfies
\begin{equation}\label{eq4:4}
    \|V_{K_0}\|\geq \delta_0>0,
\end{equation}
then
\begin{equation}\label{eq4:5}
    \Phi(\bs V)=\int_{\Omega_1}\Bigl(\DP{\bs{AV}}{\bs V}-\sum_{i=1}^n d_i\|\nabla V_i\|^2\Bigr)\D x\leq -\delta_0^2q_1(S_1),
\end{equation}
where $q_1(S_1)$ is a positive quantity that can only increase if the measure $S_1$ decreases.
\end{lemma}
\begin{proof}
First of all we note that the values of the inner product $\DP{\bs{AV}}{\bs V}$ are determined by the symmetric part of $\bs A$. Indeed, consider
$$
\bs A=\frac 12 (\bs A+\bs A^\top)+\frac 12(\bs A-\bs A^\top)=\bs A^++\bs A^-,
$$
where $\bs A^+$ is symmetric and $\bs A^-$ is skew-symmetric. Then, since $\DP{\bs{A^-V}}{\bs V}=0$,
$$
\DP{\bs{AV}}{\bs V}=\DP{\bs{A^+V}}{\bs V}.
$$

Consider the eigenvalue problem \eqref{eq1:11} in $\Omega_1$ and denote the eigenfunctions and eigenvalues $\{\psi^1_i(\bs x)\}_{i=0}^\infty$ and $\{\lambda_i^1\}_{i=0}^\infty$ respectively. From the completeness of the system of eigenfunctions it follows that any solution can be represented as in \eqref{eq4:3}, moreover
\begin{equation}\label{eq4:6}
    V_k(\bs x,t)=\sum_{j=1}^\infty c_j^k(t)\psi_j^1(\bs x),
\end{equation}
and
\begin{equation}\label{eq4:7}
    \int_{\Omega_1}V_k(\bs x,t)\D x=0,\quad k=1,\ldots,n.
\end{equation}
Let us use the equality
$$
\int_{\Omega_1}\|\nabla V_k\|^2\D x=-\int_{\Omega_1}\DP{\Delta V_k}{V_k}\D x=\sum_{j=1}^\infty \lambda_j^1c_j^k(t)c_j^s(t).
$$
Then
\begin{equation}\label{eq4:8}
    \Phi(\bs V)=\sum_{j=1}^\infty \DP{(\bs A^+-\lambda_j^1\bs D)\bs c_j(t)}{\bs c_j(t)},
\end{equation}
where $\bs c_j(t)=(c_j^1,\ldots,c_j^n(t))$. Since $\Omega_1\subset \Omega$ then $\lambda_j^1>\lambda_1$ where $\lambda_1$ is the first nonzero eigenvalue of \eqref{eq1:11} in $\Omega$. From Lemma \ref{lem:A:2} it follows that all the eigenvalues of $\bs A^+-\lambda_j^1\bs D$ will be negative and hence
$$
\DP{(\bs A^+-\lambda_j^1\bs D)\bs c_j(t)}{\bs c_j(t)}\leq q_j(S_1)\|\bs c_j(t)\|^2,\quad j=1,2,\ldots.
$$
Here $\|\bs c_j(t)\|^2=\sum_{k=1}^n \left(c_j^k(t)\right)^2$, and $q_j(S_1)$ are positive quantities such that
\begin{equation}\label{eq4:9}
    q_1(S_1)\leq q_2(S_1)\leq \ldots.
\end{equation}
Note that $q_1(S_1)$ may only increase if the measure $S_1$ of $\Omega_1$ decreases.

Now we use the condition \eqref{eq4:4} that implies that
$$
\sum_{j=1}^\infty \|\bs c_j(t)\|^2\geq \sum_{j=1}^\infty (c_j^{K_0})^2=\|V_{K_0}(\bs x,t)\|^2\geq \delta_0^2.
$$
As a result, using \eqref{eq4:9} and \eqref{eq4:8}, we obtain \eqref{eq4:5}.
\end{proof}

\begin{theorem}\label{th4:1} Let the conditions of Lemma \ref{lem4:1} hold and assume that
\begin{equation}\label{eq4:10}
    \min_{\bs w\in\bd S_n}\{\DP{\bs p}{\bs{Aw}}-\DP{\bs{Aw}}{\bs w}\}=-m<0
\end{equation}
for any $\bs p\in\Int S_n$. Then system \eqref{eq1:1}--\eqref{eq1:3} is permanent if the measure $S_1$ of $\Omega_1$ is sufficiently small.
\end{theorem}
\begin{proof}
Consider the functional
\begin{equation}\label{eq4:11}
    F(\bs v)=\exp\left(\sum_k p_k\,\overline{\log v_k(\bs x,t)}\right),\quad \bs p\in\Int S_n,
\end{equation}
defined on the solutions to \eqref{eq1:1}--\eqref{eq1:3} with the support in $\Omega_1$ and nonzero initial conditions
$$
v_k(\bs x,0)=\varphi(\bs x)>0,\quad \bs x\in\Omega_1,\quad k=1,\ldots,n.
$$
In \eqref{eq4:11}
$$
\overline{\log v_k(\bs x,t)}=\int_{\Omega_1}\log v_k(\bs x,t)\D x.
$$
Then
\begin{equation}\label{eq4:12}
    F(\bs v)|_{t=0}=\exp\left(\sum_k p_k\,\overline{\log \varphi_k(\bs x)}\right)=F_0>0.
\end{equation}
If there exists at least one solution $v_k(\bs x,t)\to 0$ as $t\to\infty$ then $F(\bs v)\to 0$. On the other hand, from \eqref{eq4:11} and the equations of system \eqref{eq1:1} it follows
\begin{equation}\label{eq4:13}
    \frac{\D F(\bs v)}{\D t}=F(\bs v)\int_{\Omega_1}\Bigl(\DP{\bs{Av}}{\bs p}-\DP{\bs A^+\bs v}{\bs v}+\sum d_k\|\nabla v_k\|^2\Bigr)\D x.
\end{equation}
Using the representation \eqref{eq4:3}, taking into account \eqref{eq4:6} and \eqref{eq4:7}, we get for \eqref{eq4:13}
$$
\frac{\D F(\bs v)}{\D t}=F(\bs v)\Bigl(\DP{\bs{Aw}}{\bs p}-\DP{\bs{Aw}}{\bs w}-\Phi(\bs V)\Bigr),
$$
where $\Phi(\bs V)$ is given by \eqref{eq4:5}.

Using \eqref{eq4:10} and inequality \eqref{eq4:5} we find
$$
\frac{\D F(\bs v)}{\D t}\geq F(\bs v)(-m+\delta_0^2q_1(S_1)).
$$
If $S_1$ is small enough then
\begin{equation}\label{eq4:14}
    q_1(S_1)\geq \frac{m}{\delta_0^2}\,,
\end{equation}
hence $F(\bs v)\geq F_0>0$ for any $t>0$, which proves the system permanence.

\end{proof}
\section{Appendix}
In the main text we use several facts about the eigenvalues of nonnegative matrices that we prove here.
\begin{lemma}\label{lem:A:1}Let $\bs A$ be a nonnegative square matrix, for which the conditions of the Perron--Frobenius theorem hold. Let $\bs D_1$ and $\bs D_2$ be two diagonal matrices such that $0<d_j^{(1)}<d_j^{(2)}$ for all $j$. Then there exist positive $\lambda_1$ and $\lambda_2$ such that
$$
\det(\bs A-\lambda_1\bs D_1)=\det (\bs A-\lambda_2\bs D_2)=0,
$$
and in particular $$\lambda_1>\lambda_2.$$
\end{lemma}
\begin{proof}Matrices $\bs B_1=\bs{AD}_1^{-1}$ and $\bs B_2=\bs{AD}_2^{-1}$ satisfy the Perron--Frobenius theorem and clearly $\bs B_1<\bs B_2$. We need to show that the spectral radius $\lambda(\bs B_1)$ of $\bs B_1$ is less than $\lambda(\bs B_2)$. But this follows from the inequality with a positive $\bs x$
$$
\bs B_2\bs x=\bs B_1\bs x+(\bs B_2-\bs B_1)\bs x>\lambda(\bs B_1)+\varepsilon \bs x
$$
and the general fact that $\bs{Ax}<\beta\bs x$ implies $\lambda(\bs A)<\beta$ (e.g., \cite{horn2012matrix}).
\end{proof}

\begin{lemma}\label{lem:A:2}Let $\bs D$ be a diagonal matrix with positive elements on the main diagonal, and let $\bs A$ be a square non-negative matrix for which the Perron--Frobenius theorem holds. If $\lambda_0$ is the dominant eigenvalue of $\bs A\bs D^{-1}$ then the eigenvalues of the matrix
$$
\bs A-\lambda\bs D
$$
with $\lambda>\lambda_0$ have negative real parts.
\end{lemma}
\begin{proof}The proof is straightforward in the case when $\bs D=\diag(d_0,\ldots,d_0)$. In this case the dominant eigenvalue $\mu$ of $\bs A$ is related to $d_0$ as $\mu=\lambda_0d_0$. Hence if $\lambda>\lambda_0$ then $\mu-\lambda d_0<0$ and all the eigenvalues of $\bs A-\lambda \bs D$ will have negative real parts. In the general case the same reasonings are used for the matrix $\bs A\bs D^{-1}$.
\end{proof}
\paragraph{Acknowledgements:} ASB and VPP are supported in part by the Russian Foundation for Basic Research grant \#13-01-00779.


\begin{thebibliography}{10}

\bibitem{aronson1986porous}
D.~G. Aronson.
\newblock The porous medium equation.
\newblock In {\em Nonlinear diffusion problems}, pages 1--46. Springer, 1986.

\bibitem{barenblatt1989theory}
G.~I. Barenblatt, V.~M. Entov, and V.~M. Ryzhik.
\newblock {\em Theory of fluid flows through natural rocks}.
\newblock Kluwer Academic Publishers, 1989.

\bibitem{bertsch1990discontinuous}
M.~Bertsch, R.~Dal~Passo, and M.~Ughi.
\newblock Discontinuous "viscosity" solutions of a degenerate parabolic
  equation.
\newblock {\em Transactions of the American Mathematical Society},
  320(2):779--798, 1990.

\bibitem{bratus2016diffusive}
A.~S. Bratus, C.-K. Hu, M.~V. Safro, and A.~S. Novozhilov.
\newblock On diffusive stability of eigen's quasispecies model.
\newblock {\em Journal of Dynamical and Control Systems}, 22(1):1--14, 2016.

\bibitem{bratus2006ssc}
A.~S. Bratus and V.~P. Posvyanskii.
\newblock {Stationary solutions in a closed distributed Eigen--Schuster
  evolution system}.
\newblock {\em Differential Equations}, 42(12):1762--1774, 2006.

\bibitem{bratus2009existence}
A.~S. Bratus, V.~P. Posvyanskii, and A.~S. Novozhilov.
\newblock {Existence and stability of stationary solutions to spatially
  extended autocatalytic and hypercyclic systems under global regulation and
  with nonlinear growth rates}.
\newblock {\em Nonlinear Analysis: Real World Applications}, 11:1897--1917,
  2010.

\bibitem{bratus2011}
A.~S. Bratus, V.~P. Posvyanskii, and A.~S. Novozhilov.
\newblock {A note on the replicator equation with explicit space and global
  regulation}.
\newblock {\em Mathematical Biosciences and Engineering}, 8(3):659--676, 2011.

\bibitem{novozhilov2013replicator}
A.~S. Bratus, V.~P. Posvyanskii, and Novozhilov~A. S.
\newblock Replicator equations and space.
\newblock {\em Mathematical Modelling of Natural Phenomena}, 9(3):47--67, 2014.

\bibitem{cantrell2003spatial}
R.~S. Cantrell and C.~Cosner.
\newblock {\em {Spatial ecology via reaction-diffusion equations}}.
\newblock Wiley, 2003.

\bibitem{cressman1987density}
R.~Cressman and A.~T. Dash.
\newblock {Density dependence and evolutionary stable strategies}.
\newblock {\em Journal of Theoretical Biology}, 126(4):393--406, 1987.

\bibitem{cressman1997sad}
R.~Cressman and G.~T. Vickers.
\newblock {Spatial and Density Effects in Evolutionary Game Theory}.
\newblock {\em Journal of Theoretical Biology}, 184(4):359--369, 1997.

\bibitem{dal1987degenerate}
R.~Dal~Passo and S.~Luckhaus.
\newblock A degenerate diffusion problem not in divergence form.
\newblock {\em Journal of differential equations}, 69(1):1--14, 1987.

\bibitem{dieckmann2000}
U.~Dieckmann, R.~Law, and J.~A.~J. Metz.
\newblock {\em {The Geometry of Ecological Interactions: Simplifying Spatial
  Complexity}}.
\newblock {Cambridge University Press}, 2000.

\bibitem{evans_2010}
L.~C. Evans.
\newblock {\em {Partial Differential Equations}}.
\newblock {American Mathematical Society}, 2nd edition, 2010.

\bibitem{hofbauer1998ega}
J.~Hofbauer and K.~Sigmund.
\newblock {\em {Evolutionary Games and Population Dynamics}}.
\newblock Cambridge University Press, 1998.

\bibitem{hofbauer2003egd}
J.~Hofbauer and K.~Sigmund.
\newblock {Evolutionary game dynamics}.
\newblock {\em Bulletin of American Mathematical Society}, 40(4):479--519,
  2003.

\bibitem{horn2012matrix}
R.~A. Horn and C.~R. Johnson.
\newblock {\em Matrix analysis}.
\newblock Cambridge university press, 2012.

\bibitem{hutson1992travelling}
V.~C.~L. Hutson and G.~T. Vickers.
\newblock {Travelling waves and dominance of ESS's}.
\newblock {\em Journal of Mathematical Biology}, 30(5):457--471, 1992.

\bibitem{hutson1995sst}
V.~C.~L. Hutson and G.~T. Vickers.
\newblock {The Spatial Struggle of Tit-For-Tat and Defect}.
\newblock {\em Philosophical Transactions of the Royal Society. Series B:
  Biological Sciences}, 348(1326):393--404, 1995.

\bibitem{knabner2003numerical}
P.~Knabner and L.~Angerman.
\newblock {\em Numerical methods for elliptic and parabolic partial
  differential equations}, volume~44.
\newblock Springer, 2003.

\bibitem{lewontin1969meaning}
R.~C. Lewontin.
\newblock The meaning of stability.
\newblock In {\em Brookhaven symposia in biology}, volume~22, pages 13--24,
  1969.

\bibitem{lieberman2005evolutionary}
E.~Lieberman, C.~Hauert, and M.~.A Nowak.
\newblock Evolutionary dynamics on graphs.
\newblock {\em Nature}, 433(7023):312--316, 2005.

\bibitem{novozhilov2012reaction}
A.~S. Novozhilov, V.~P. Posvyanskii, and A.~S. Bratus.
\newblock On the reaction--diffusion replicator systems: spatial patterns and
  asymptotic behaviour.
\newblock {\em Russian Journal of Numerical Analysis and Mathematical
  Modelling}, 26(6):555--564, 2012.

\bibitem{postnikov2007continuum}
E.~B. Postnikov and I.~M. Sokolov.
\newblock Continuum description of a contact infection spread in a sir model.
\newblock {\em Mathematical biosciences}, 208(1):205--215, 2007.

\bibitem{schreiber2013spatial}
S.~J. Schreiber and T.~P. Killingback.
\newblock Spatial heterogeneity promotes coexistence of rock--paper--scissors
  metacommunities.
\newblock {\em Theoretical population biology}, 86:1--11, 2013.

\bibitem{schuster1983rd}
P.~Schuster and K.~Sigmund.
\newblock {Replicator dynamics}.
\newblock {\em Journal of Theoretical Biology}, 100:533--538, 1983.

\bibitem{smith2011dynamical}
H.~L. Smith and H.~R. Thieme.
\newblock {\em Dynamical systems and population persistence}, volume 118 of
  {\em Graduate Studies in Mathematics}.
\newblock American Mathematical Society Providence, RI, 2011.

\bibitem{vaidya2012spontaneous}
N.~Vaidya, M.~L. Manapat, I.~A. Chen, R.~Xulvi-Brunet, E.~J. Hayden, and
  N.~Lehman.
\newblock Spontaneous network formation among cooperative rna replicators.
\newblock {\em Nature}, 491(7422):72--77, 2012.

\bibitem{vickers1989spa}
G.~T. Vickers.
\newblock {Spatial patterns and ESS's.}
\newblock {\em Journal of Theoretical Biology}, 140(1):129--35, 1989.

\bibitem{Vickers1991}
G.~T. Vickers.
\newblock Spatial patterns and travelling waves in population genetics.
\newblock {\em Journal of Theoretical Biology}, 150(3):329--337, Jun 1991.

\bibitem{weinberger1991ssa}
E.~D. Weinberger.
\newblock {Spatial stability analysis of Eigen's quasispecies model and the
  less than five membered hypercycle under global population regulation}.
\newblock {\em Bulletin of Mathematical Biology}, 53(4):623--638, 1991.

\end{thebibliography}

\end{document}